\documentclass[letterpaper, 10 pt, conference]{ieeeconf}  
\usepackage{comment}
\usepackage[backend=biber,style=ieee,maxbibnames=6]{biblatex}
\usepackage{amsmath,amssymb,amsfonts}
\usepackage{tabularx}
\PassOptionsToPackage{dvipsnames,svgnames,table}{xcolor}
\usepackage{multirow}
\usepackage{graphicx}
\usepackage{textcomp}
\usepackage{xcolor}
\usepackage{array}
\usepackage{booktabs}
\usepackage{subcaption}
\usepackage{graphicx}
\usepackage{placeins}        
\usepackage{stfloats}        

\usepackage{algorithm}
\usepackage{algpseudocode}
\usepackage{hyperref}

\usepackage{tikz}
\usetikzlibrary{arrows.meta,positioning,calc}

\usepackage{tabularx}
\usepackage{booktabs}
\newtheorem{remark}{Remark}
\newtheorem{theorem}{Theorem}

\usepackage{capt-of}

\newcolumntype{P}[1]{>{\centering\arraybackslash}p{#1}}
\def\BibTeX{{\rm B\kern-.05em{\sc i\kern-.025em b}\kern-.08em
    T\kern-.1667em\lower.7ex\hbox{E}\kern-.125emX}}

\begin{document}

\title{\LARGE \bf
Model-Free Budgeted Attack Scheduling for Cyber-Physical Systems}

\author{%
Qazi~Mairaj~ud~din$^{1}$,
Sidra~Ghayour~Bhatti$^{2}$,
and Qadeer~Ahmed$^{3}$%
\thanks{$^{1}$Qazi Mairaj ud din is with the Department of Electrical and Computer Engineering and the Center for Automotive Research, The Ohio State University, Columbus, OH 43212, USA.
{\tt\small mairajuddin.1@osu.edu}}%
\thanks{$^{2}$Sidra Ghayour Bhatti is with the Center for Automotive Research, The Ohio State University, Columbus, OH 43212, USA.
{\tt\small bhatti.39@osu.edu}}%
\thanks{$^{3}$Qadeer Ahmed is with the Department of Mechanical and Aerospace Engineering and the Center for Automotive Research, The Ohio State University, Columbus, OH 43212, USA.
{\tt\small ahmed.358@osu.edu}}%
}

\maketitle
\thispagestyle{empty}
\pagestyle{empty}

\begin{abstract}
This letter studies the budgeted scheduling of stealthy false
data-injection (FDI) attacks against state estimators in cyber–physical systems. Existing event-based attack
schedulers requires full knowledge of the plant model and assume the residual distribution is exactly Gaussian---assumptions that fail for real-world CPS sensor streams whose residuals are heavy-tailed and whose dynamics are unknown to the adversary. We propose a model-free attack-scheduler that replaces the parametric Gaussian threshold with the empirical 
quantile of a learned sequence autoencoder residual, calibrated from measurements alone without any plant matrices. We prove that the realized attack rate converges almost surely to the target budget 
under stationary ergodic residuals. Experiments on two synthetic systems and a real heavy-duty truck dataset show that the proposed scheduler tracks the budget to within $1$--$2\%$ while also preserving the residual magnitude, guaranteeing stealthiness against any residual-based detector. Comparing with the model-based baseline---granted the true plant and innovation covariance---mis-realizes the budget by up to $8.96\%$ under
heavy-tailed residual distribution causing the attacker to achieve only
$1.37\times$ system degradation when $1.84\times$ is intended.
\end{abstract}


\section{Introduction}
{\bfseries C}yber--physical systems (CPSs)
are vulnerable to cyber attacks, which are broadly categorized as denial-of-service (DoS) attacks and deception attacks \cite{zhou2024cybersecurity}. Among deception attacks, stealthy false data-injection (FDI) attacks are particularly destructive as they evade residual-based anomaly detectors while degrading system performance. Analyzing the worst-case FDI attack is therefore essential to the design of effective countermeasures.

A practical adversary operates under resource constraints and cannot attack at every instant. More formally, in practical adversarial setting, constraint of a \emph{resource budget}, where the communication-bandwidth and energy limitations are considered to design \emph{attack rate budget} $\bar{\Gamma}\in(0,1)$: the adversary may transmit an attack signal on at most a fraction $\bar{\Gamma}$ of the time steps over the deployment horizon. This motivates \emph{event-based} scheduling, in which the action is taken only after an informative event occurs. Such events carry useful information about the system and event-triggered policies have been shown to effectively trade communication rate against estimation quality \cite{wu2012event,han2015stochastic}. This principle was first developed for benign sensor data scheduling, where Wu~\emph{et al.}~\cite{wu2012event} proposed scheduling whenever the Mahalanobis-normalized innovation exceeds a threshold. The same idea was later turned to adversarial use: energy-constrained DoS attack scheduling determines when to jam the channel to maximize the estimation error \cite{zhang2015optimal, zhang2016optimal}. DoS attacks, however, only block transmissions and are detectable,
whereas stealthy FDI attacks corrupt the data while remaining
hidden \cite{zhou2024cybersecurity}. Recent works, \cite{cheng2019event, Guo2023, Wang2025}, accordingly schedule FDI attacks under joint stealthiness and budget constraints: Guo \emph{et al.} \cite{Guo2023} introduce an event-based optimal stealthy FDI scheme that triggers attack based on the residual information (event-based) and maximizes the trace of the estimation error covariance, and Wang \emph{et al.} \cite{Wang2025} extend event-based FDI
scheduling to resource-constrained multi-sensor systems. More generally, the event-based attack scheduling rule whitens the true Kalman innovation $z_k$ using the steady-state innovation covariance $S$ and computes the triggering threshold by inverting the Gaussian tail probability to satisfy the prescribed attack budget. Equivalently, they assume the attacker has full knowledge of the plant model $(A,C,Q,R)$ and that the innovation is exactly Gaussian.

These two assumptions are not realistic in practice. Real-world CPS applications like vehicle sensor streams---wheel speed, engine speed---produce innovation residuals that are markedly heavy-tailed arising from
nonlinearities, asymmetric operational constraints, 
\& intermittent events \cite{zhao2021event,Guo2026adaptive}, rather than Gaussian. Furthermore, an external adversary may only have access to the sensor stream and does not possess the plant knowledge needed to compute $S$. To the best of our knowledge, no existing attack scheduler realizes a target stealthy-attack budget without a plant model and without a
Gaussianity assumption.This letter addresses these gaps with contributions as follows:
\begin{enumerate}
\item We propose a \emph{model-free}, non-parametric attack-scheduler that fires when a learned residual from a pretrained sequence autoencoder exceeds the empirical $(1-\bar\Gamma)$ quantile of a calibration stream, requiring no plant matrices and no distributional assumption on the residual (Section~\ref{sec:trigger}).
\item We prove that the realized attack rate converges almost
surely to the target budget $\bar\Gamma$ under stationary ergodic
residuals, with calibration error decaying as $O(N^{-1/2})$ in the calibration window $N$ (Theorem~\ref{thm:rate}).
\item We show that the resulting schedule is stealthy by
construction: the attack injection preserves the residual
magnitude, leaving any 
residual-based detector's false-alarm rate
unchanged at every budget without requiring Gaussianity (Section~\ref{sec:stealthiness}).
\item We validate the scheduler on two synthetic systems (Gaussian and heavy-tailed non-Gaussian noise) and a real heavy-duty truck dataset \cite{biggs2024modeling}: the proposed trigger
achieves $1$--$2\%$ budget tracking across
$\bar\Gamma\in[0.02,0.5]$; comparing with the model-based baseline
which mis-realizes the budget by up to $8.96\%$ under heavy-tailed
noise---a budget-dependent error that reduces attack
effectiveness from the intended $1.84\times$ to only
$1.37\times$ system degradation at $\bar\Gamma=0.30$
(Section~\ref{sec:experiments}).
\end{enumerate}

\emph{Notation:} $\mathbb{E}[\cdot]$ denotes expectation and
$\mathrm{tr}(\cdot)$ the trace. For $x\in\mathbb{R}^n$, $\|x\|$ is
the Euclidean norm. $\mathcal{N}(\mu,\Sigma)$ denotes the Gaussian
distribution with mean $\mu$ and covariance $\Sigma$.
$\mathbf{1}\{\cdot\}$ is the indicator function, and
$\hat{F}_N$ the empirical CDF from $N$ samples.
$I_n$ is the $n\times n$ identity matrix.



\section{Problem Setup}
\label{sec:setup}
Consider a discrete-time linear plant
\begin{equation}
x_{k+1} = A x_k + w_k, \qquad y_k = C x_k + v_k,
\label{eq:plant}
\end{equation}
with state $x_k\!\in\!\mathbb{R}^n$ and measurement
$y_k\!\in\!\mathbb{R}^p$. The process noise $w_k\sim\mathcal{N}(0,Q)$
and sensor noise $v_k$ are mutually independent; $v_k$ is
zero-mean with covariance $R$ but is \emph{not} required to be
Gaussian (see Remark~\ref{rem:nongauss}). A steady-state Kalman
filter produces the prior estimate $\hat{x}_{k|k-1}=A\hat{x}_{k-1|k-1}$  and the corresponding measurement prediction $\hat{y}_{k\mid k-1}$. The innovation
\begin{equation}
z_k = y_k - C\hat{x}_{k|k-1},
\label{eq:innov}
\end{equation}
is the standard residual signal monitored by anomaly detectors with steady-state covariance $S = C\bar{P}C^{\!\top}+R$,  where $\bar{P}$ is the steady-state prior error covariance obtained from the discrete algebraic Riccati equation. Under Gaussian $v_k$, $z_k\sim\mathcal{N}(0,S)$; under non-Gaussian $v_k$ the covariance $S$ is still well-defined but the distribution of $z_k$ departs from Gaussian.

\begin{remark}\label{rem:nongauss}
The Kalman filter gain and $S$ depend on $Q$ and $R$ only
through their second moments and remain well-defined regardless
of the distribution of $v_k$. The proposed attack scheduler
makes no distributional assumption on $v_k$ or $z_k$.
\end{remark}

\subsection*{Residual-based detection}
A standard $\chi^2$ detector raises an alarm at step $k$ when
\begin{equation}
g_k = z_k^{\!\top}S^{-1}z_k > \eta,
\label{eq:detector}
\end{equation}
where the threshold $\eta$ is set to achieve a target false-alarm
rate $\epsilon_d$ (typically $1\%$) on nominal data
\cite{Guo2023}.

\subsection*{Threat model}
The adversary operates in a \emph{measurement-only} setting:
it observes a nominal stream $\{y_k\}_{k=1}^{N}$ during an
eavesdropping phase and can inject signals into the sensor-to-estimator
network channel during the attack window. It has no access to $A,C,Q,R$,
observer internals (gains or covariances), or detector parameters.
For a fair comparison with the model-based baseline
\cite{cheng2019event,Guo2023,Wang2025}, we adopt the worst-case sign-flip
attack from \cite[Thm.~3]{Guo2023} as the action at firing
instants:
\begin{equation}
y_k^c =
\begin{cases}
2\,\hat{y}_{k|k-1} - y_k, & \gamma_k = 1,\\[2pt]
y_k, & \gamma_k = 0,
\end{cases}
\label{eq:signflip}
\end{equation}
where $\gamma_k\!\in\!\{0,1\}$ is the attack trigger.
Substituting \eqref{eq:signflip} into \eqref{eq:innov} gives the compromised innovation $z_k^{c}=y_k^{c}-\hat{y}_{k\mid k-1}=-z_k$, which is the sign-flipped nominal residual and preserves the second-order
statistics of $z_k$ exactly, therefore evading any
residue-based detector \cite{greenwood1996guide}.
Since $\hat{y}_{k|k-1}$ is unavailable in the model-free
setting, we replace it with a learned predictor introduced in
Section~\ref{sec:trigger}.

\subsection*{Budgeted attack scheduling problem}
Under a resource budget $\bar\Gamma\!\in\!(0,1)$, the adversary
requires the long-run attack rate to satisfy
\begin{equation}
\lim_{T\to\infty}\frac{1}{T}\sum_{k=1}^{T}\gamma_k
\;\leq\; \bar\Gamma.
\label{eq:budget}
\end{equation}
where $T$ is the deployment horizon. The scheduling problem is to design a trigger $\gamma_k$ that is a measurable function of the observable stream $\{y_j\}_{j\leq k}$ such that \eqref{eq:budget} holds,
and the attacks concentrate on instants of large nominal residual (events) to maximize estimation degradation while remaining detector-stealthy, and also requiring neither $A,C,Q,R$ nor Gaussianity of the innovation.

\section{Budgeted Attack Scheduling Framework}
\label{sec:baseline}
The model-based attack schedulers proposed in \cite{cheng2019event,Guo2023,Wang2025} adopt an event-triggered scheduler which first whitens the innovation by its steady-state covariance, 
\begin{equation}
\zeta_k \;=\; S^{-1/2}\,z_k,
\label{eq:whiten}
\end{equation}
Here $S^{-1/2}$ is the symmetric inverse square root of the
innovation covariance, so that $\zeta_k$ has identity
covariance whenever $S$ matches the true innovation
covariance. The scheduler triggers attacks when its sup-norm exceeds a threshold $\varepsilon$:
\begin{equation}
\gamma_k \;=\;
\begin{cases}
1 & \text{if } \|\zeta_k\|_\infty > \varepsilon,\\
0 & \text{otherwise}
\end{cases}
\quad\text{\cite[Eq.~(16)]{Guo2023}}.
\label{eq:guo-trigger}
\end{equation}
Under the additional assumption that $\zeta_k\!\sim\!\mathcal{N}(0,I_m)$ has independent standard normal entries, \cite[Thm.~1]{Guo2023} establishes the closed-form budget-to-threshold map
\begin{equation}
\Gamma \;=\; 1-\bigl(1-2Q(\varepsilon)\bigr)^{m},
\label{eq:gauss-rate}
\end{equation}
where $Q(\cdot)$ is the Gaussian tail probability and $m=\dim(z_k)$ is the dimension of the (vector-valued) residual
on the attacked channel. Equation~\eqref{eq:gauss-rate} follows from the
fact that, under Gaussianity, the probability that all $m$
entries of $\zeta_k$ satisfy $|\zeta_k^{(j)}|\leq\varepsilon$
is $(1-2Q(\varepsilon))^m$. Inverting~\eqref{eq:gauss-rate}
gives the threshold for a prescribed budget $\bar\Gamma$:
\begin{equation}
\varepsilon^{\dagger}
\;=\; Q^{-1}\!\Bigl(\tfrac{1}{2}
\bigl[1-(1-\bar\Gamma)^{1/m}\bigr]\Bigr).
\label{eq:gauss-thresh}
\end{equation}

\subsection*{Limitations of the model-based framework}
Implementing \eqref{eq:whiten}--\eqref{eq:gauss-thresh}
requires two assumptions that are unrealistic in adversarial
settings. First, computing $S^{-1/2}$ requires the plant
matrices $(A,C,Q,R)$ via the discrete algebraic Riccati
equation---knowledge unavailable to a realistic measurement-only
adversary. Second, the budget-rate map \eqref{eq:gauss-rate}
requires $\zeta_k$ to be exactly standard normal with
independent entries. As noted by~\cite{han2015stochastic}, the deterministic threshold \eqref{eq:guo-trigger} itself distorts the conditional distribution of the innovation, so the Gaussian assumption is self-contradictory at the operating point. Beyond this structural issue, real CPS residuals routinely depart from Gaussianity: plant
nonlinearities introduce higher-order moments into the
innovation; asymmetric operational constraints such as actuator
saturation truncate one tail of the residual distribution;
intermittent events such as mode switches and scheduled
reconfigurations produce transient deviations; and sensor
outliers, quantization errors, and packet drops generate
occasional large innovations~\cite{han2015stochastic,karlgaard2008adaptive,
pearson2002outliers}. The next section proposes a
non-parametric attack-scheduler that requires neither plant knowledge nor
Gaussianity.


\section{Proposed Model-Free Trigger}
\label{sec:trigger}
We propose a model-free approach where we replace both ingredients of \eqref{eq:whiten} the whitened innovation \& \eqref{eq:gauss-thresh} the
Gaussian-tail inverse by a learned manifold residual score using historic measurements and derived empirical quantile.

\subsection{Learned predictor and manifold residual score}
Nominal measurements of the LTI system \eqref{eq:plant} lie
near the \emph{measurement manifold}
$\mathcal{M}=\operatorname{col}(C)\subset\mathbb{R}^p$: in the
noiseless limit, $y_k=Cx_k\in\mathcal{M}$, so the measurement
covariance $\Sigma_y = C\Sigma_x C^{\!\top}+R$ has a signal
subspace equal to $\operatorname{col}(C)$ \cite{viberg1995subspace}.
Subspace identification methods exploit precisely this
structure: by estimating $\Sigma_y$ from data, the column
space of $C$ can be recovered without knowledge of $A,C,Q,R$
\cite{viberg1995subspace,kim2014subspace}. We adopt the same philosophy
for the adversarial setting.

Let $\mathrm{AE}^{\star}$ be a sequence autoencoder with a
bottleneck of dimension $q\le\operatorname{rank}(C\Sigma_x
C^{\!\top})$ 
, trained on nominal data by
minimizing the reconstruction MSE
\begin{equation}
\mathcal{L}(\mathrm{AE}) \;=\;
\mathbb{E}\bigl[\bigl\|y_k - \mathrm{AE}(y_k)\bigr\|^2\bigr].
\label{eq:ae-loss}
\end{equation}
The theoretical role of $\mathrm{AE}^{\star}$ follows from two
complementary results. For a \emph{linear} autoencoder,
Baldi et al. \cite{BaldiHornik1989} proved that
$\mathcal{L}$ has a unique global minimum (up to rotation in
the latent space) with all other critical points being saddle
points; at this minimum,
\begin{equation}
\mathrm{AE}^{\star}(y) \;=\; P_U\,y,
\label{eq:linear-ae}
\end{equation}
where $P_U$ is the orthogonal projection onto the subspace $U$
spanned by the top-$q$ eigenvectors of
$\mathbb{E}[y_k y_k^{\!\top}]$. In the small-noise regime of
\eqref{eq:plant}, the signal subspace $U$ converges to $\operatorname{col}(C)$, so $s_k = \|y_k - P_U y_k\|_2$ is exactly
the distance from $y_k$ to the measurement manifold
$\mathcal{M}$.

For a \emph{nonlinear} (deep or regularized) autoencoder,
Alain et al.\cite{alain2014regularized} showed that
minimizing a contractive or denoising form of
\eqref{eq:ae-loss} forces $\mathrm{AE}^{\star}$ to be
insensitive to variations \emph{orthogonal} to the local
manifold while remaining sensitive to variations \emph{along}
it. More precisely, the reconstruction function
$r(y)=\mathrm{AE}^{\star}(y)$ estimates the score
(gradient of the log-density) of the nominal data
distribution, with the Jacobian $\partial r/\partial y$ having
small singular values in directions off the manifold and
order-one singular values along manifold tangents. As a consequence, the
reconstruction error $r(y)-y$ points toward the nearest
high-density region on the data manifold, so that
$s_k = \|y_k - \mathrm{AE}^{\star}(y_k)\|_2$ is small
near the nominal measurement manifold $\mathcal{M}$ and
large whenever $y_k$ deviates from $\mathcal{M}$,
regardless of the specific noise distribution. The
identification $\mathrm{AE}^{\star}\approx P_U$ holds locally
under Lipschitz continuity of the encoder--decoder pair.

This theoretical grounding motivates the use of $s_k$ as a
plant-agnostic anomaly score: high values of $s_k$ signal that
$y_k$ has deviated from the nominal measurement manifold, which
is exactly the condition under which the model-based trigger
\eqref{eq:guo-trigger} also fires ($\|\zeta_k\|_\infty >
\varepsilon$). The shared firing geometry is what allows the
empirical-quantile threshold below to reproduce the budget
behavior of \eqref{eq:gauss-thresh} without requiring
$A,C,Q,R$ or Gaussianity. Furthermore, by learning the
measurement subspace from data alone the autoencoder enables the adversary to
operate under the measurement-only threat model of
Section~\ref{sec:setup} without acquiring any system
parameters.

Define the per-step manifold residual score
\begin{equation}
s_k \;=\; \bigl\| y_k - \mathrm{AE}^{\star}(y_k)\bigr\|_{2},
\label{eq:ae-score}
\end{equation}
which serves both as the scheduling score in
\eqref{eq:mf-trigger} and, after sign-flip, as the basis for
the measurement prediction
$\hat{y}_{k|k-1} \approx \mathrm{AE}^{\star}(y_k)$ required
to implement \eqref{eq:signflip} in the model-free setting.

\subsection{Empirical-quantile threshold}
Given a budget $\bar{\Gamma}$ and a calibration stream of
$N$ nominal samples, let
\begin{equation}
\widehat{F}_N(t)\;=\;\frac{1}{N}\sum_{k=1}^{N}\mathbf{1}\{s_k\le t\}
\label{eq:emp-cdf}
\end{equation}
be the empirical CDF of the score \cite{vd1998asymptotic} ---  i.e., the fraction of nominal samples whose score does not exceed $t$. The trigger threshold is the generalized inverse of $\widehat{F}_N$ at level $1-\bar{\Gamma}$,
\begin{equation}
\hat{\varepsilon}_N
\;=\;\inf\bigl\{t\in\mathbb{R}:\widehat{F}_N(t)\ge 1-\bar{\Gamma}\bigr\},
\label{eq:mf-thresh}
\end{equation}

i.e., the smallest threshold for which at least a fraction
$1-\bar{\Gamma}$ of the calibration scores fall below it. Calibrating detection or scheduling thresholds as empirical residual quantiles is a standard
non-parametric approach in statistical analysis and CPS anomaly detection
\cite{chandola2009anomaly, murguia2020security}, providing a
fixed nominal false-trigger rate without distributional
assumptions on $s_k$. By
construction, \eqref{eq:mf-thresh} sets $\hat{\varepsilon}_N$
at the $(1-\bar{\Gamma})$-quantile of the nominal score
distribution, so that a fraction $\bar{\Gamma}$ of nominal
scores exceed it. The trigger is then
\begin{equation}
\gamma_k \;=\;
\begin{cases}
1 & \text{if } s_k > \hat{\varepsilon}_N,\\
0 & \text{otherwise.}
\end{cases}
\label{eq:mf-trigger}
\end{equation}
Construction \eqref{eq:ae-score}--\eqref{eq:mf-trigger}
requires neither $A,C,Q,R$ nor any assumption on the marginal
distribution of $s_k$. It depends on the nominal stream only
through $\widehat{F}_N$, and on the budget only through the
quantile level $1-\bar{\Gamma}$. 

This non-parametric calibration replaces
\cite[Thm.~1]{Guo2023} and \cite[Thm.~1]{Wang2025}. It does
not require Gaussianity, it does not need a closed-form tail,
and it tolerates the heavy-tailed residuals. The price of non-parametric calibration is that the budget guarantee is asymptotic rather than closed-form; the model-based map \eqref{eq:gauss-rate} is exact only under
Gaussianity, which fails in practice, whereas Theorem~\ref{thm:rate} establishes almost-sure convergence of $\hat\Gamma_{N,T}\to\bar\Gamma$ under ergodicity alone and quantifies the $O(N^{-1/2})$ finite-sample gap.

\begin{remark}[Plug-in Gaussian baseline]
A natural attacker-feasible variant of the trigger scheduler in section \ref{sec:baseline} \eqref{eq:whiten}--\eqref{eq:gauss-thresh} replaces the unknown $S$ with the empirical covariance
\begin{equation}
\widetilde{S}_N \;=\; \frac{1}{N}\sum_{k=1}^{N} r_k\,r_k^{\!\top},
\qquad r_k \;=\; y_k - \mathrm{AE}^{\star}(y_k)\in\mathbb{R}^p,
\label{eq:S-tilde}
\end{equation}
computed offline from the AE residuals on the calibration
stream. 
The trigger statistic becomes
$\tilde{\zeta}_k = \widetilde{S}_N^{-1/2} r_k$, and the
threshold is set by the Gaussian-tail inversion
\eqref{eq:gauss-thresh} with $S$ replaced by $\widetilde{S}_N$:
\begin{equation}
\tilde{\gamma}_k \;=\;
\begin{cases}
1 & \text{if } \!\bigl\{\|\widetilde{S}_N^{-1/2} r_k\|_\infty>\varepsilon^{\dagger}\bigr\},\\
0 & \text{otherwise}
\end{cases}
\label{eq:plugin-trigger}
\end{equation}

This plug-in construction keeps the Gaussian-tail inverse
\eqref{eq:gauss-thresh} but only resolves the unknown-$S$
issue, not the non-Gaussianity issue; we use it as one of the
comparison baselines in Section~\ref{sec:experiments}.
\end{remark}

\section{Convergence of the Realized Rate}
\label{sec:theory}
Since with our proposed model-free approach, non-parametric construction sacrifices the closed-form
of \eqref{eq:gauss-rate}, so we must verify that the realized
attack rate still tracks the budget. The following result
establishes almost-sure convergence under conditions strictly
weaker than those required for \eqref{eq:gauss-rate}.

\begin{theorem}
\label{thm:rate}
Let $\{s_k\}_{k\ge 1}$ be a stationary, ergodic, real-valued
process with continuous marginal CDF $F$. Assume the
calibration stream (of length $N$, used to form $\hat{\varepsilon}_N$) and the
deployment stream (of length $T$, over which the rate is realized)
are generated by the same marginal law and are mutually
independent. For a target budget
$\bar{\Gamma}\in(0,1)$, define $\hat{\varepsilon}_N$ by
\eqref{eq:mf-thresh} and the realized rate over horizon $T$,
\begin{equation}
\hat{\Gamma}_{N,T}
\;=\; \frac{1}{T}\sum_{k=1}^{T}\mathbf{1}\!\bigl\{s_k>\hat{\varepsilon}_N\bigr\}.
\label{eq:realized-rate}
\end{equation}
Then $\hat{\Gamma}_{N,T}\to\bar{\Gamma}$ almost surely as
$N\to\infty$ and $T\to\infty$.
\end{theorem}

\begin{proof}
We split the argument into the convergence of the calibrated
threshold and that of the deployment-time average at a
fixed threshold along with their link to each other.

\emph{Step 1 (threshold consistency).} By the generalized
Glivenko--Cantelli theorem for stationary ergodic sequences
\cite{tucker1959generalization}
,
\begin{equation}
\sup_{t\in\mathbb{R}}\bigl|\widehat{F}_N(t)-F(t)\bigr|
\;\xrightarrow[\text{a.s.}]{N\to\infty}\;0.
\label{eq:gc}
\end{equation}
Let $\varepsilon^{\star}=F^{-1}(1-\bar{\Gamma})$ and fix any
$\delta>0$. Since $F$ is continuous and strictly increasing in a
neighborhood of $\varepsilon^{\star}$, there exists $\eta>0$ with
$F(\varepsilon^{\star}+\delta)\ge 1-\bar{\Gamma}+\eta$ and
$F(\varepsilon^{\star}-\delta)\le 1-\bar{\Gamma}-\eta$. On the
a.s.\ event where the supremum in \eqref{eq:gc} is below $\eta$,
the definition \eqref{eq:mf-thresh} forces
$|\hat{\varepsilon}_N-\varepsilon^{\star}|\le\delta$. Hence
\begin{equation}
\hat{\varepsilon}_N
\;\xrightarrow[\text{a.s.}]{N\to\infty}\; \varepsilon^{\star}.
\label{eq:eps-consistency}
\end{equation}

\emph{Step 2 (population rate at the true threshold).}
The ergodic theorem applied to the bounded function
$\mathbf{1}\{s_k>\varepsilon^{\star}\}$ gives
\begin{equation}
\frac{1}{T}\sum_{k=1}^{T}\mathbf{1}\!\bigl\{s_k>\varepsilon^{\star}\bigr\}
\;\xrightarrow[\text{a.s.}]{T\to\infty}\;
1-F(\varepsilon^{\star})
\;=\; \bar{\Gamma}.
\label{eq:pop-rate}
\end{equation}

\emph{Step 3 (joint limit).} For any $\delta>0$, on the a.s.\
event $\{|\hat{\varepsilon}_N-\varepsilon^{\star}|<\delta\}$,
which by \eqref{eq:eps-consistency} holds for all sufficiently
large $N$, the inclusions
$\{s_k>\varepsilon^{\star}+\delta\}\subseteq\{s_k>\hat{\varepsilon}_N\}\subseteq\{s_k>\varepsilon^{\star}-\delta\}$
yield
\begin{equation}
\frac{1}{T}\sum_{k=1}^{T}\mathbf{1}\!\bigl\{s_k>\varepsilon^{\star}+\delta\bigr\}
\;\le\; \hat{\Gamma}_{N,T}
\;\le\;
\frac{1}{T}\sum_{k=1}^{T}\mathbf{1}\!\bigl\{s_k>\varepsilon^{\star}-\delta\bigr\}.
\label{eq:sandwich}
\end{equation}
Sending $T\to\infty$ and applying the ergodic theorem to both
bounds gives
\begin{equation}
1-F(\varepsilon^{\star}+\delta)
\;\le\; \liminf_{N,T\to\infty}\hat{\Gamma}_{N,T}
\;\le\; \limsup_{N,T\to\infty}\hat{\Gamma}_{N,T}
\;\le\; 1-F(\varepsilon^{\star}-\delta),
\label{eq:liminflimsup}
\end{equation}
almost surely. Letting $\delta\to 0$ and invoking continuity
of $F$ at $\varepsilon^{\star}$ yields
$\hat{\Gamma}_{N,T}\to\bar{\Gamma}$ a.s.
\end{proof}

\begin{remark}[Finite-sample budget error]
Theorem~\ref{thm:rate} is asymptotic. Under standard mixing conditions on $\{s_k\}$, and provided the
marginal density $f$ exists and is positive at
$\varepsilon^{\star}$, the Bahadur representation for sample
quantiles of dependent sequences \cite{vd1998asymptotic} together
with the central limit theorem for stationary mixing sequences
yields
\begin{equation}
\hat{\Gamma}_{N,T} - \bar{\Gamma}
\;=\; O_p\bigl(N^{-1/2}\bigr) + O_p\bigl(T^{-1/2}\bigr),
\label{eq:rate-clt}
\end{equation}
decomposing the finite-sample budget error into a calibration
component, controlled by the eavesdropping window $N$, and a
deployment-time stochastic component, controlled by the horizon
$T$. The empirical $O(N^{-1/2})$ decay is verified in
Section~\ref{sec:experiments} (Figure.~\ref{fig:convergence}).
\end{remark}

\begin{remark}[Comparison with the Gaussian-tail calibration]
The map \eqref{eq:gauss-rate}--\eqref{eq:gauss-thresh} 
is exact only when $\zeta_k$ is
Gaussian. If a misspecified covariance $\widetilde{S}\neq S$
is used, $\widetilde{S}^{-1/2}z_k$ is not whitened, and the
realized rate under threshold $\varepsilon^{\dagger}$ from
\eqref{eq:gauss-thresh} can deviate from $\bar{\Gamma}$ by an
amount that does not vanish as $N\to\infty$. By contrast,
Theorem~\ref{thm:rate} requires only stationarity and ergodicity
of $\{s_k\}$, and the realized rate converges to $\bar{\Gamma}$
regardless of the marginal distribution.
\end{remark}

\section{Experimental Validation}
\label{sec:experiments}
We evaluate the proposed model-free attack scheduler along four axes:
(i) realized-rate accuracy versus target budget;
(ii) convergence with the calibration window;
(iii) comparison against the model-based baseline under both gaussian and
non-gaussian distribution; and (iv) stealth preservation against a
residual-based detector. All experiments use the residual
sign-flip attack \eqref{eq:signflip} at firing instants.

Two numerical examples and a real-world vehicle CAN, \textbf{CSU truck} dataset \cite{biggs2024modeling} is used to evaluate and compare the proposed model-free attack-scheduler against model-based baseline \cite{cheng2019event,Guo2023,Wang2025}. The plant for both
examples is a fully known linear plant
\begin{equation}
A=\begin{bmatrix}0.95&0.02\\0&0.90\end{bmatrix},\quad
C=\begin{bmatrix}1&0\end{bmatrix},\quad
Q=0.01\,I_2,\quad R=0.05,
\label{eq:num-plant}
\end{equation}
stable with $\rho(A)=0.95$. This way the model-based baseline can be granted full plant knowledge (A,C,Q,R), enabling a fair head-to-head comparison. For the real-world dataset, plant matrices are unavailable so the model-based
baseline cannot be applied. This is precisely the setting that motivates the model-free approach operating from measurements alone, depicting more realistic adversarial settings. 

\textbf{Experiment~1} uses Gaussian sensor noise
$v_k\!\sim\!\mathcal{N}(0,R)$ with 4th moment, i.e. excess kurtosis $\kappa_v
\!=\!0$ (no heavy-tail). \textbf{Experiment~2} uses the
Gaussian mixture \eqref{eq:mix-noise}, that share the \emph{same}
marginal variance $R$, so that the model-based baseline
receives an identical steady-state covariance $S$ in both,
isolating distributional shape as the sole variable between
experiments. This simulates non-Gaussianity and heavy-tailed behavior of real-world CPS sensors while inducing innovation excess kurtosis $\kappa_z\!=\!7.5$ (heavy-tail). The Gaussianity assumption underlying \eqref{eq:gauss-thresh} holds exactly in the former and fails in the latter.
\begin{equation}
v_k \;\sim\; (1\!-\!p)\,\mathcal{N}(0,\sigma_1^2)
\;+\; p\,\mathcal{N}(0,\sigma_2^2),
\label{eq:mix-noise}
\end{equation}
with $p\!=\!0.05$, $\sigma_1^2\!=\!0.0227$,
$\sigma_2^2\!=\!0.5682$. 
A separate sequence autoencoder $\mathrm{AE}^{\star}$ model is trained on each dataset using a single-layer LSTM with hidden dimension $128$, latent dimension $1$, and window length $L=50$, trained by
reconstruction MSE on a $70\%$ training split; thresholds are
calibrated on the subsequent $15\%$ and all rates are evaluated on the
held-out final $15\%$. Calibration residuals for the CSU dataset are trimmed at $5\sigma$ to exclude non-stationary driving transients ($0.65\%$ of steps removed).

\subsubsection{Budget tracking}
Table~\ref{tab:budget} reports the realized rate $\hat\Gamma_{N,T}$ against the target budget $\bar\Gamma$ and their corresponding errors for all experiments. As expected, in Experiment~1 both methods track the budget to within $1\%$ across $\bar\Gamma\in[0.02,0.5]$, confirming that the proposed scheduler matches the model-based baseline when the Gaussian assumption holds. In Experiment~2, however, Figure.~\ref{fig:budget-exp2}  shows that the model-based baseline---despite having access to the true plant knowledge, correct $S$ and the true innovation---deviates from the target by up to $8.96\%$, over-firing at small budgets and under-firing at larger ones, whereas the proposed scheduler tracks the budget to a mean error of $0.24\%$. The non-vanishing gap arises because the whitened innovation $z_k/\sqrt{S}$ is not standard normal under the heavy-tailed mixture noise, so the Gaussian-tail threshold is mis-placed regardless of the estimation of $S$. Since the
plant matrices are unknown, the model-based baseline cannot be evaluated for the CSU truck dataset (hence "N/A" Table~\ref{tab:budget}). However, the proposed trigger tracks the budget on the truck dataset to a mean
error of $1.49\%$ and a maximum error of $2.18\%$ across
$\bar\Gamma\in[0.02,0.5]$, without access to any plant
information. This accuracy is comparable to the model-based baseline on Experiment~2 (Table~\ref{tab:budget}), where the Gaussian assumption fails,
demonstrating that the model-free scheduler generalises from the
controlled synthetic setting to a real non-Gaussian vehicle data.

\begin{figure}[t]
\centering
\includegraphics[width=0.8\linewidth]{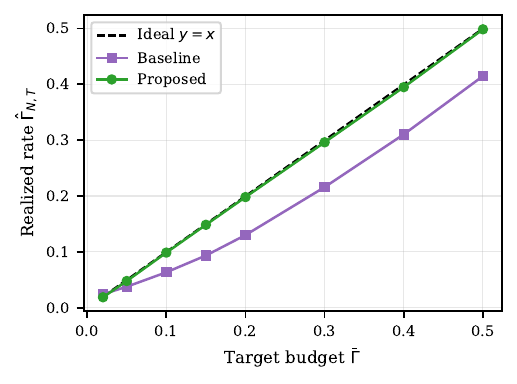}
\caption{Realized attack rate $\hat\Gamma_{N,T}$ versus target
budget $\bar\Gamma$ for Experiment~2 (non-Gaussian).Experiment~1
(Gaussian) and CSU truck results are reported in Table~\ref{tab:budget}.}
\label{fig:budget-exp2}
\end{figure}

\begin{table}[t]
\centering
\caption{$\hat{\Gamma}_{N,T}$ error across all datasets.}
\label{tab:budget}
\setlength{\tabcolsep}{5pt}
\begin{tabular}{lcccc}
\toprule
Dataset & \multicolumn{2}{c}{Mean $|$err$|$} & \multicolumn{2}{c}{Max $|$err$|$}\\
\cmidrule(lr){2-3}\cmidrule(lr){4-5}
 & Baseline & Proposed & Baseline & Proposed\\
\midrule
Exp.\ 1 (Gaussian) & $0.23\%$ & $0.54\%$ & $0.5\%$ & $1.0\%$\\
Exp.\ 2 (non-Gaussian)  & $5.48\%$ & $0.24\%$ & $8.96\%$ & $0.51\%$\\
CSU Truck (Real)   & N/A & $1.48\%$ & N/A & $2.18\%$\\
\bottomrule
\end{tabular}
\end{table}

\subsubsection{Convergence with the calibration window}
Figure.~\ref{fig:convergence} plots the budget error
$|\hat\Gamma_{N,T}-\bar\Gamma|$ against the calibration window $N$ at
$\bar\Gamma=0.1$ for experiment 2. The error decays from
$1.63\%$ at $N=200$ to $0.52\%$ at $N=17{,}350$, following the
$O(N^{-1/2})$ reference, empirically corroborating
Theorem~\ref{thm:rate}.

\begin{figure}[t]
\centering
\includegraphics[width=0.8\linewidth]{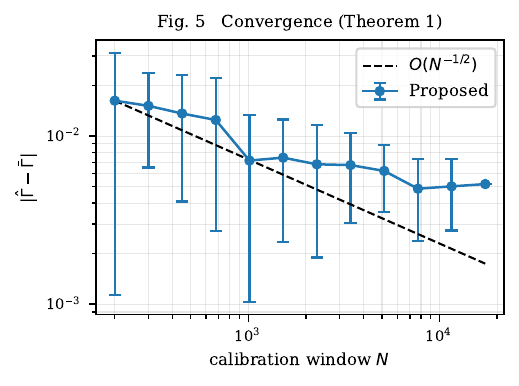}
\caption{Budget error $|\hat\Gamma_{N,T}-\bar\Gamma|$ versus
calibration window $N$}
\label{fig:convergence}
\end{figure}

\subsubsection{Performance degradation}\label{sec:degradation}
We evaluate the impact of the attack \eqref{eq:signflip} on the estimator via the trace of the state-estimation error covariance (which measures how much the estimator performance is degraded) via $\mathrm{tr}(P^a_k)\approx\mathbb{E}[\|e^a_k\|^2]$, estimated by Monte-Carlo averaging over $50$ independent noise realizations. The attack is activated after an initial $3{,}200$-step nominal window, allowing the filter to reach steady state before attack onset. The nominal (no-attack) trace is $\mathrm{tr}(P)=0.069$ for both experiments. In both experiments, the trace rises monotonically after attack onset over the deployment horizon for the attack rate $\bar\Gamma\in\{0,0.10,0.30,0.50\}$, confirming that more attack resources cause greater performance degradation. In Experiment~1 both schedulers produce nearly identical degradation at each budget, since both realize the target rate accurately on Gaussian data. However, in Experiment~2 (Figure.~\ref{fig:trpa}), the baseline under-fires on non-Gaussian data and achieves less degradation (1.358x at $\bar\Gamma\in\{0.30\}$) than intended, translating directly into a loss of attack effectiveness. The proposed trigger fires at the correct rate
and delivers the intended impact (1.835x at $\bar\Gamma\in\{0.30\}$) at every budget. 
The performance-degradation cannot be evaluated on the CSU truck
dataset since it requires the true state $x_k$ to compute
$\mathrm{tr}(P^a_k) = \mathbb{E}[\|x_k - \hat{x}^a_k\|^2]$,
which is not available without a plant model. However, the AE-residual excess kurtosis of the truck data ($\kappa_r=5.33$) closely matches Experiment~2, confirming that the Experiment~2 is a faithful statistical proxy for the heavy-tailed residuals observed in real vehicle data. 
\begin{figure}[t]
\centering
\includegraphics[width=1\linewidth]{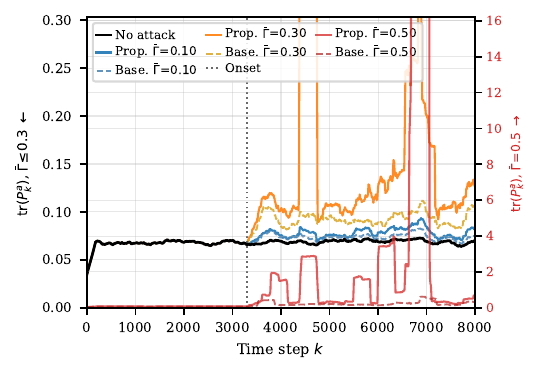}
\caption{$\mathrm{tr}(P^a_k)$ under attack for Experiment~2
(non-Gaussian) at $\bar\Gamma\in\{0,0.10,0.30,0.50\}$,
Left axis ($\bar\Gamma\!=\!0, 0.10, 0.30$); right axis ($\bar\Gamma\!=\!0.50$,
red). Experiment~1 results (identical degradation for both schedulers)
are omitted for brevity.}
\label{fig:trpa}
\end{figure}

\subsubsection{Attack stealthiness}\label{sec:stealthiness}
The sign-flip \eqref{eq:signflip} preserves $|z_k|$ exactly,
so the windowed $\chi^2$ detection index (window $N\!=\!20$ and
threshold $\vartheta$ at $1\%$ nominal false-alarm rate)
$g_t=\sum_{k=t-N+1}^{t}z_k^{\!\top}S^{-1}z_k$ is
identically distributed under attack and no attack. This is
confirmed empirically: the attacked false-alarm rate equals
the nominal baseline of $1.08\%$ (Experiment~1) and $0.69\%$
(Experiment~2) at every budget
$\bar\Gamma\in\{0.10,0.30,0.50\}$. For the CSU truck dataset,
stealthiness follows analytically from the identity
$|z^c_k|\!=\!|z_k|$, which holds regardless of noise
distribution; empirical verification is not possible since
$z_k$ and $S$ require plant knowledge unavailable for real
CAN data.

\section{Conclusion}
We have proposed a non-parametric, plant-agnostic
attack-scheduling trigger for resource-budgeted stealthy
attacks on cyber-physical systems. The trigger replaces model-based dependencies to whiten-innovation sup-norm, with
the residual of a pretrained sequence autoencoder, and
replaces the Gaussian-tail closed form with an empirical
quantile of the residual on a nominal stream. We have shown
that the realized attack rate converges almost surely to the
target budget under stationary ergodicity, without any
Gaussianity assumption and without any knowledge of the plant
matrices, and our experimental design is structured to verify
both the asymptotic and finite-sample properties of this
guarantee on real-world vehicle data. Combining the proposed
trigger with the manifold-consistent injection and budgeted
sensor-selection components of our broader framework is the
subject of future work.


\renewcommand*{\bibfont}{\scriptsize}

\printbibliography

\newpage

\appendix


\section{System Architecture and Trigger Illustration}
\label{app:system}

The model-based scheduler proceeds in three steps. First, it 
whitens the innovation by the known covariance $S$ to produce 
$\zeta_k$ with unit-variance components. Second, it fires 
whenever any component of $\zeta_k$ is large. Third, it sets 
the firing threshold by inverting the Gaussian tail formula: 
given a budget $\bar\Gamma$, Eq.~\eqref{eq:gauss-thresh} gives 
the exact threshold assuming $\zeta_k$ is standard Gaussian. 
All three steps require plant knowledge and Gaussianity.

Figure.~\ref{fig:attack-diagram} illustrates the overall
system architecture for model-free
attack scheduler. The plant generates state $x_k$;
the sensor produces measurement $y_k$. The model-free
attack scheduler observes $y_k$ and sets the trigger
$\gamma_k\in\{0,1\}$: when $\gamma_k=1$ the corrupted
measurement $y_k^c$ \eqref{eq:signflip} is injected into
the network in place of $y_k$. The remote estimator
produces $\hat{x}_k$ from whatever measurement it
receives, and the detector raises an alarm when
$g_k>\eta$. The scheduler has no access to the estimator
internals, the detector threshold, or any plant matrices.

\begin{figure}[H]
\centering
\includegraphics[width=\linewidth]{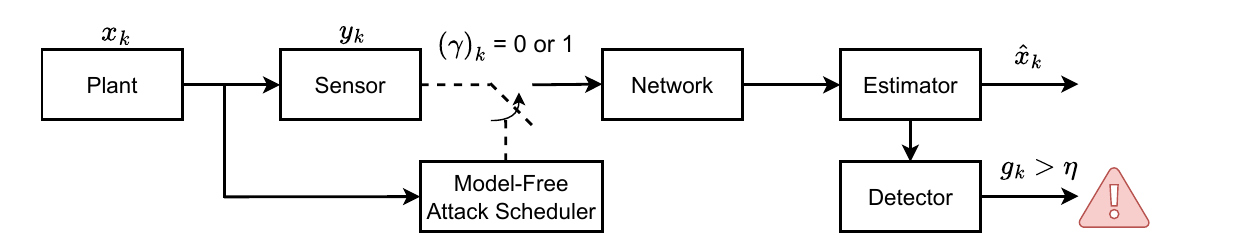}
\caption{System architecture. The model-free attack
scheduler intercepts the sensor measurement $y_k$ and
decides $\gamma_k\in\{0,1\}$ using only the AE residual
score $s_k$. When $\gamma_k=1$ the sign-flipped
measurement $y_k^c$ is transmitted; the estimator and
detector operate on the received signal without knowledge
of the attack.}
\label{fig:attack-diagram}
\end{figure}

The proposed method replaces all three steps with a single 
non-parametric operation, illustrated in Figure.~\ref{fig:cal_deploy}: 
collect $N$ nominal AE scores $\{s_k\}$ (panel~(a), grey 
histogram), set the threshold $\hat\varepsilon_N$ at the 
$(1-\bar\Gamma)$ empirical quantile (dashed line, shaded area 
$= \bar\Gamma$), and fire at every deployment step where 
$s_k > \hat\varepsilon_N$ (panel~(b), orange dots). No $S$, 
no Gaussian tail, no plant matrices. Geometric interpretation of the model-free attack-scheduler is illustrated in figure\ref{fig:AE}.
\begin{figure}
    \centering
        \includegraphics[width=1\linewidth]{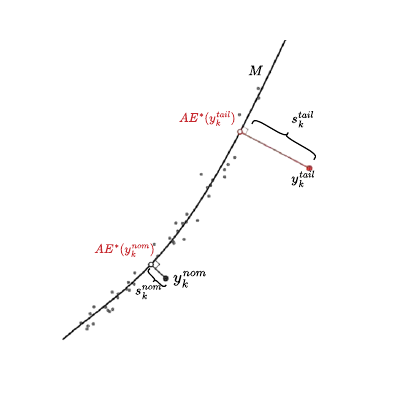}
        \caption{Geometric interpretation of the model-free attack-scheduler.
        Nominal measurements $y_k^{\text{nom}}$ lie near the
        measurement manifold $\mathcal{M}=\operatorname{col}(C)$
        and are reconstructed with small residual by
        $\mathrm{AE}^{\star}$. An attacked measurement
        $y_k^{\text{att}}$ deviates from $\mathcal{M}$, producing
        a large score $s_k=\|y_k-\mathrm{AE}^{\star}(y_k)\|_2$
        that triggers the scheduler.}
        \label{fig:AE}
\end{figure}

\begin{figure}
    \centering
    \includegraphics[width=0.75\linewidth]{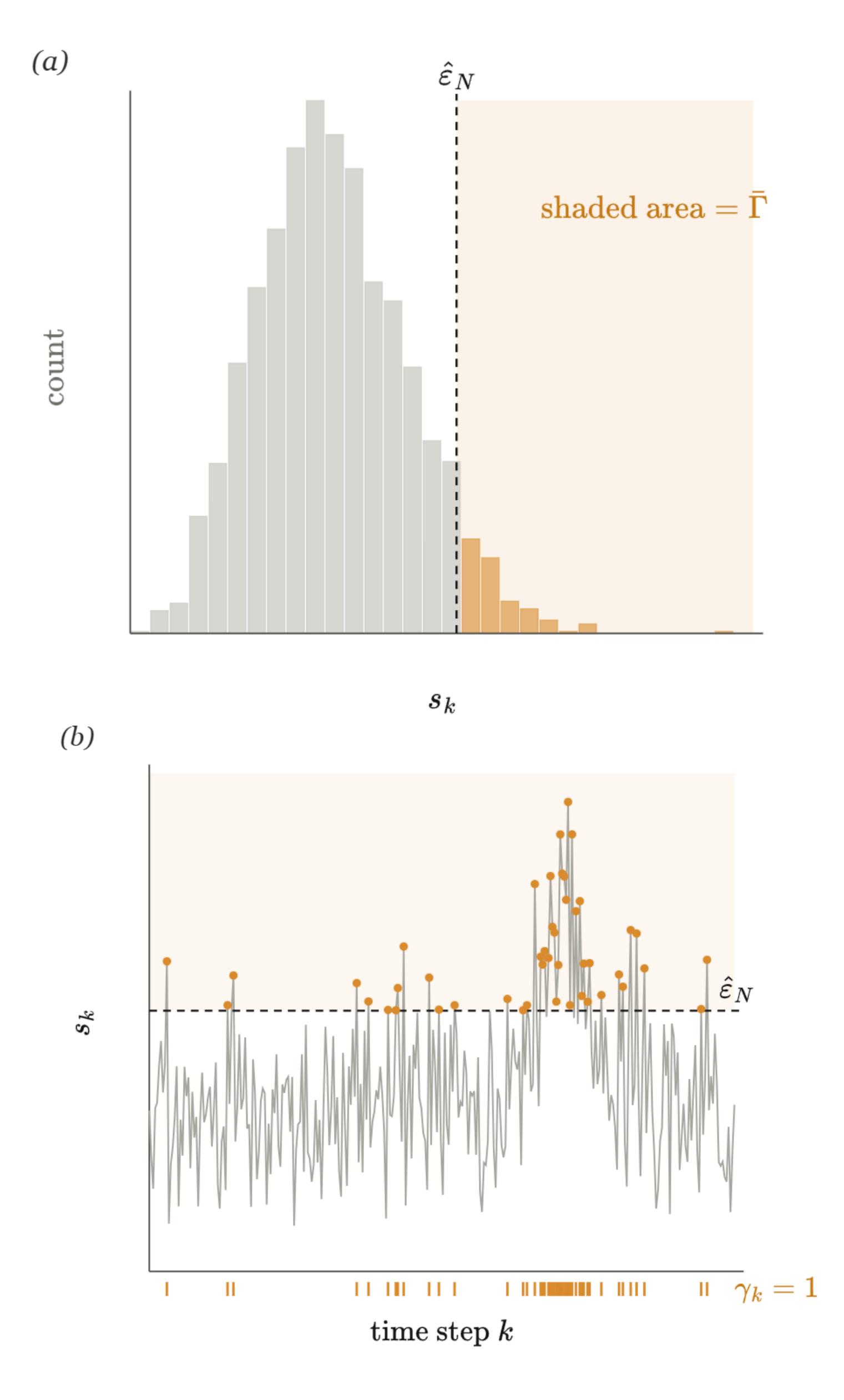}
    \caption{Empirical-quantile trigger calibration and deployment.
    (a) The threshold $\hat\varepsilon_N$ is set as the
    $(1-\bar\Gamma)$-quantile of the score $s_k$ on a held-out
    nominal calibration stream: by construction, the area to its
    right is $\bar\Gamma$. (b) At deployment, the trigger fires on
    the time steps where $s_k>\hat\varepsilon_N$ (orange markers),
    shown beneath the time axis as $\gamma_k=1$.}
    \label{fig:cal_deploy}
\end{figure}

\begin{remark}[Bottleneck dimension and manifold rank]
The choice $q=1$ is not a heuristic but follows from the
system structure: the measurement manifold
$\mathcal{M}=\operatorname{col}(C)$ has dimension equal to
$\operatorname{rank}(C)=1$ for the scalar output channel of
\eqref{eq:num-plant}. Setting $q=\operatorname{rank}(C)$
ensures that $\mathrm{AE}^{\star}$ recovers exactly the
signal subspace and that $s_k$ measures the orthogonal
distance from $y_k$ to $\mathcal{M}$. For multi-channel
systems ($p > 1$), $q$ should be set to
$\operatorname{rank}(C)$, which can be estimated from the
singular value decomposition of the sample output covariance
without plant knowledge.
\end{remark}

Equation~\eqref{eq:emp-cdf} counts, for each candidate threshold 
$t$, what fraction of the $N$ calibration scores fall at or below 
$t$. This gives a staircase function that climbs from $0$ to $1$ 
as $t$ increases. Equation~\eqref{eq:mf-thresh} then asks: at 
what value of $t$ does this staircase first reach 
$1\!-\!\bar\Gamma$? That value is the threshold $\hat\varepsilon_N$. 
Since exactly a fraction $\bar\Gamma$ of calibration scores exceed 
$\hat\varepsilon_N$ by construction, and the deployment scores come 
from the same distribution, the trigger fires at approximately 
$\bar\Gamma$ of deployment steps---without ever assuming Gaussianity 
of the score distribution.

\section{AE Residual Distribution Statistics}
\label{app:ae-stats}

Table~\ref{tab:ae-full} and Figure.~\ref{fig:ae-histograms}
report the autoencoder residual statistics and distribution
plots for all four datasets. Experiment~1 exhibits
near-Gaussian residuals ($\kappa_r\!=\!0.06$, QQ-plot on
the reference line throughout). Experiment~2 and both real
vehicle datasets show pronounced heavy tails
($\kappa_r\!=\!5.3$--$6.8$), with QQ-plot points departing
from the reference line at both extremes. The similar excess
kurtosis across Experiment~2, the CSU truck, and the KIA
Soul confirms that the synthetic mixture is a faithful proxy
for real vehicle sensor noise, justifying its use as a
controlled surrogate where no plant model is available.

\begin{table}
\centering
\caption{AE residual distribution statistics (calibration
stream; real datasets after $5\sigma$ trimming).}
\label{tab:ae-full}
\setlength{\tabcolsep}{4pt}
\begin{tabular}{lrrrrc}
\toprule
Dataset & Mean & Std & Skew & $\kappa_r$ & Heavy-tailed?\\
\midrule
Exp.~1 (Gaussian) & $0.004$ & $0.145$ & $ 0.01$ & $0.06$ & No\\
Exp.~2 (mixture)  & $0.008$ & $0.153$ & $ 0.16$ & $5.81$ & Yes\\
CSU Truck         & $0.003$ & $0.047$ & $-0.33$ & $5.33$ & Yes\\
KIA Soul          & $-0.062$ & $1.605$ & $0.05$ & $6.77$ & Yes\\
\bottomrule
\multicolumn{6}{l}{\footnotesize $\kappa_r$: AE-residual excess
kurtosis. Trimmed at $5\sigma$: CSU $0.65\%$, KIA ${<}1\%$.}
\end{tabular}
\end{table}

\begin{figure}
\centering
\includegraphics[width=\linewidth]{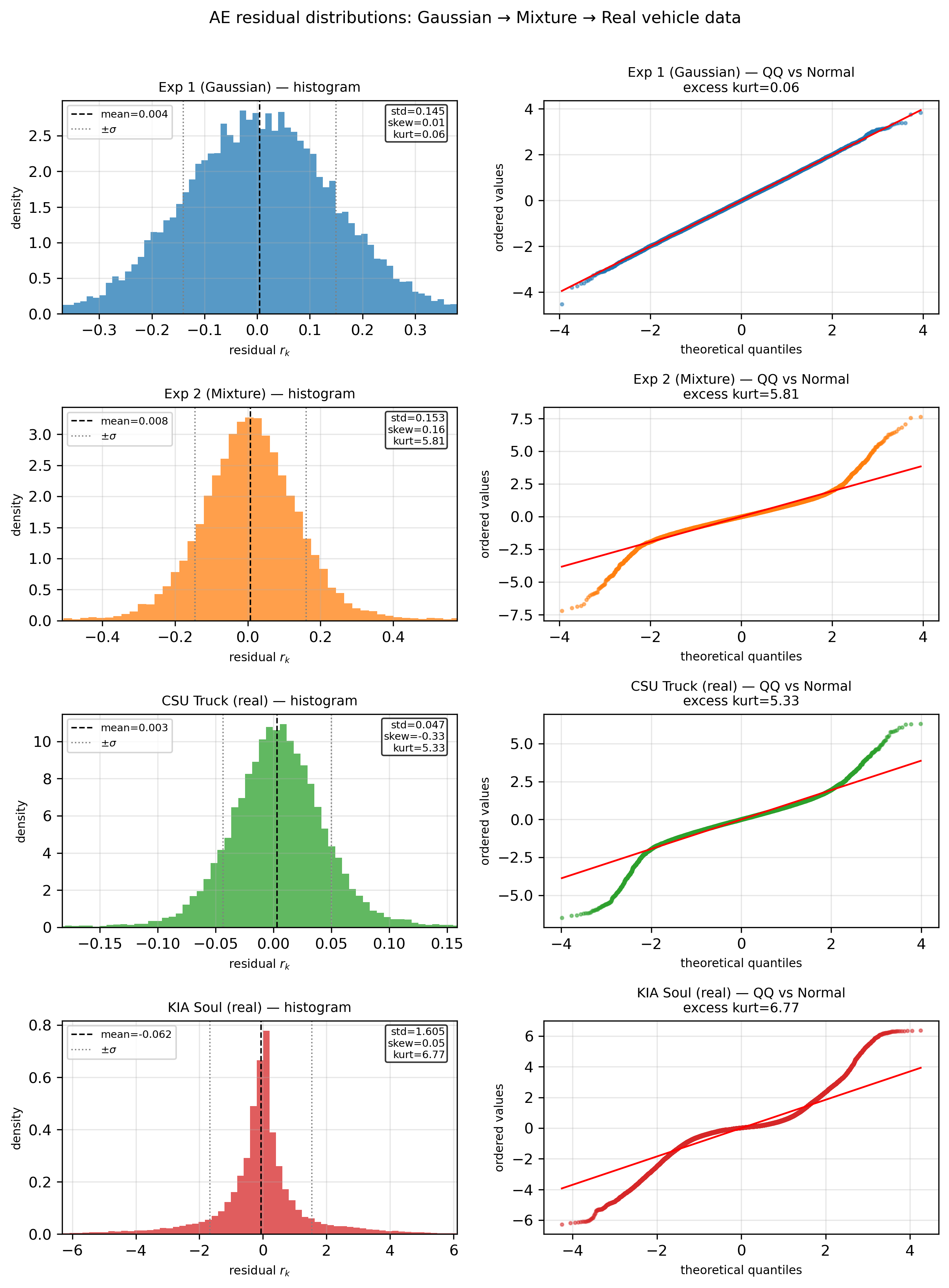}
\caption{AE residual histograms (left) and normal QQ-plots
(right) for all four datasets. Experiment~1 follows the
Gaussian reference closely. Experiment~2, CSU truck, and
KIA Soul all exhibit heavy-tailed departure from normality
with similar excess kurtosis ($5.3$--$6.8$), supporting
Experiment~2 as a synthetic proxy for real vehicle data.}
\label{fig:ae-histograms}
\end{figure}

\section{State Trajectory Deviation}
\label{app:state-dev}

Figure.~\ref{fig:x1-deviation} shows the state-estimate
trajectory $\hat{x}_1$ under the sign-flip attack for
both experiments (one panel per budget). In Experiment~1
the attacked and nominal trajectories diverge modestly and
similarly for both schedulers. In Experiment~2 the
deviation grows with the budget, reflecting the greater
degradation delivered by the proposed scheduler compared
to the under-firing baseline.

\begin{figure}
\centering
\subfloat[Experiment~1 (Gaussian)\label{fig:x1-a}]{%
  \includegraphics[width=\linewidth]{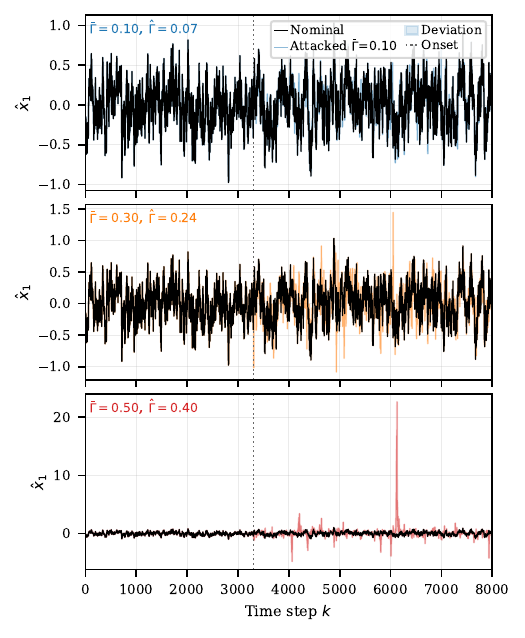}}
\\[3pt]
\subfloat[Experiment~2 (Mixture)\label{fig:x1-b}]{%
  \includegraphics[width=\linewidth]{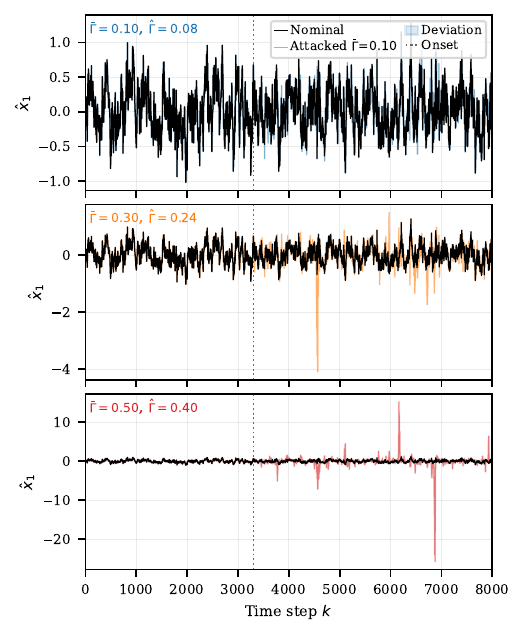}}
\caption{$\hat{x}_1$ under attack (one panel per budget).
Black: nominal. Coloured: attacked. Shaded: deviation from
nominal. Dotted: attack onset. Budget and realized rate
annotated top-left of each panel.}
\label{fig:x1-deviation}
\end{figure}

\section{Windowed $\chi^2$ Detector Time Series}
\label{app:detector}

Figure.~\ref{fig:detector} shows the windowed $\chi^2$
detection index $g_t$ ($N\!=\!20$) under no attack and
under attack at $\bar\Gamma\!=\!0.30$ for both experiments.
The attacked and nominal traces are statistically
indistinguishable throughout the deployment horizon; the
attacked false-alarm rate matches the nominal baseline
($1.08\%$ for Experiment~1, $0.69\%$ for Experiment~2)
to four decimal places at every budget, confirming
stealthiness.

\begin{figure}
\centering
\subfloat[Experiment~1 (Gaussian), $\vartheta=37.3$%
         \label{fig:det-a}]{%
  \includegraphics[width=\linewidth]{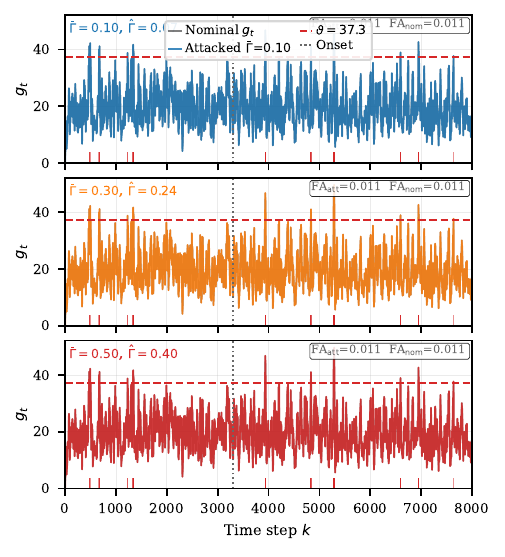}}
\\[2pt]
\subfloat[Experiment~2 (Mixture), $\vartheta=81.8$%
         \label{fig:det-b}]{%
  \includegraphics[width=\linewidth]{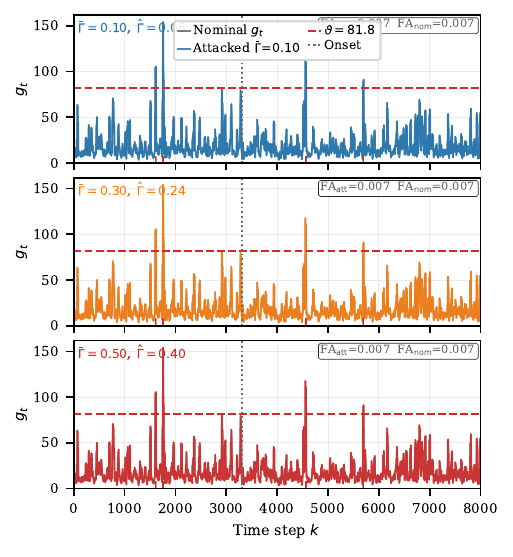}}
\caption{Windowed $\chi^2$ detection index for
each target budget $\bar\Gamma\in\{0.10,0.30,0.50\}$. Red ticks at the
base of each panel mark individual alarms.}
\label{fig:detector}
\end{figure}

\end{document}